\documentclass[11pt]{article}
\usepackage[utf8]{inputenc}
\usepackage{amsfonts}
\usepackage{braket}
\usepackage[bottom]{footmisc}
\usepackage{amsmath}
\usepackage{amsthm}
\usepackage{amssymb}
\usepackage{comment}
\usepackage{subcaption}
\newtheorem{lem}{Lemma}
\newtheorem{remark}{Remark}
\newtheorem{example}{Example}

\newtheorem{defn}{Definition}
\newtheorem{fact}{Fact}
\newtheorem{claim}{Claim}
\newtheorem*{claim*}{Claim}
\newtheorem{corollary}{Corollary}
\newenvironment{proofs}{\noindent\textit{Proof Sketch:}}{\hfill$\square$}
\usepackage{mathtools}

\DeclarePairedDelimiter\floor{\lfloor}{\rfloor}
\newtheorem{theorem}{Theorem}
\usepackage{geometry}
\title{The $aBc$ Problem and Equator Sampling R\'enyi Divergences}
\author{H. Klauck\footnote{Centre for Quantum Technologies, Singapore, {\tt hklauck@gmail.com}. This work is funded by the Singapore Ministry of Education and by the Singapore National Research Foundation. Also supported by Majulab UMI 3654.} \ and
 D. Lim\footnote{Centre for Quantum Technologies, {\tt e0352761@u.nus.edu.sg}}}

\begin{document}

\date{}
\maketitle

\begin{abstract}
We investigate the problem of approximating the product $a^TBc$, where $a,c\in S^{n-1}$ and $B\in O_n$, in models of communication complexity and streaming algorithms. The worst meaningful approximation is to simply decide whether the product is 1 or -1, given the promise that it is either. We call that problem the $aBc$ problem. This is a modification of computing approximate inner products, by allowing a basis change. While very efficient streaming algorithms and one-way communication protocols are known for simple inner products (approximating $a^Tc$) we show that no efficient one-way protocols/streaming algorithms exist for the $aBc$ problem. In communication complexity we consider the 3-player number-in-hand model.
We consider a setting where the players holding $B,c$ may confer over many rounds, while there is only one message to Alice.
Our main tools for lower bounds are geometric concentration results about R\'enyi divergences.

We show that:
\begin{enumerate}
\item In communication complexity $a^TBc$ can be approximated within additive error $\epsilon$ with communication $O(\sqrt n/\epsilon^2)$ by a one-way protocol Charlie to Bob to Alice.
\item The $aBc$ problem has a streaming algorithm that uses space $O(\sqrt n \log n)$
\item Any one-way communication protocol for $aBc$ needs communication at least $\Omega(n^{1/3})$, and we prove a tight results regarding a communication tradeoff: if Charlie and Bob communicate over many rounds such that Charlie communicates $o(n^{2/3})$ and Bob $o(n^{1/3})$, and then the transcript is sent to Alice, the error will be large.
\item To establish our lower bound we show concentration results for R\'enyi divergences under the event of restricting a density function on the sphere to a random equator and subsequently normalizing the restricted density function. This extends previous results by Klartag and Regev \cite{KR} for set sizes to R\'enyi divergences of arbitrary density functions.
\item We show a strong concentration result for conditional R\'enyi divergences on bipartite systems for all $\alpha>1$, which does not hold for $\alpha=1$.
\end{enumerate}

\end{abstract}

\section{Introduction}

Inner products are fundamental in linear algebra, mapping two vectors to a number measuring their overlap. While computing and approximating such simple inner products has been investigated extensively in communication complexity and data-streaming (see e.g. \cite{KNR,ams:frequency}), more general bilinear forms have rarely been considered. In this paper we consider the complexity of computing a bilinear form/inner product of the form $a^TBc$, where $a,c$ are unit vectors and $B$ is an orthogonal matrix (all things being real). The models we consider are 3-player number-in-hand multiparty communication complexity and streaming algorithms.

A reasonably strong requirement on such a computation is to find the value of $a^TBc$ up to some {\bf additive} error $\epsilon$.
If we skip the matrix $B$, then there exist both streaming algorithms and communication protocols that require space/communication $O(\log n/\epsilon^2)$ \cite{ams:frequency}. This is also optimal \cite{ChaReg}. We call such inner products {\em simple}.

In this paper we consider non-simple inner products of the type $a^TBc$. Our main result is that this problem is hard, even for the worst meaningful approximation, in the one-way communication complexity model, and in a stronger setting we describe below. We also show that an additive approximation is possible in one-way communication complexity with communication $O(\sqrt n/\epsilon^2)$ , and that the worst meaningful approximation is computable within space $O(\sqrt n\log n)$ in the streaming model. We conjecture that an additive error approximation is also possible as efficiently in the streaming model.

One main contribution of this paper are concentration inequalities for R\'enyi divergences of density functions on the sphere, when restricted to a random equator and then normalized. These generalize a result by Klartag and Regev \cite{KR} that can be summarized as follows: given a large enough subset of the sphere, restricted to a random equator that set will be of almost the same size except with low probability. The generalization is from subsets and their sizes to (re-normalized) density functions and their divergence.

Another motivation for the $aBc$ problem is that we have previously \cite{KLim} introduced the $ABC$ problem, in which Alice, Bob, Charlie receive matrices from the special orthogonal group $SO_n$ each ($n$ even), and have to decide whether $ABC=I$ or $ABC=-I$. An algorithm for $aBc$ clearly allows us to solve the $ABC$ problem: pick any row/column $i$ and restrict the $ABC$ problem to $A_{i,\cdot}BC_{\cdot,i}$. The interest in the $ABC$ problem is that it can be solved in a model of quantum communication complexity where the whole quantum storage contains only one fully coherent qubit, and the rest of the quantum storage is in a totally mixed state (at the beginning), while requiring only $O(\log n)$ communication (on $n\times n$ matrices). A large lower bound for the randomized complexity of $ABC$ would lead to a quantum supremacy result in which a model with weak quantum storage (but good control) would outperform the corresponding classical model in a way that is provable without any assumptions and works for errors as large as constant.
We note here that the $aBc$ problem itself allows a quantum number-in-hand one-way protocol of complexity $O(\log n)$, in which the final measurement is against an observable defined by a single quantum state (the problems considered in \cite{KR, expsep} require a measurement against an observable defined by an $n/2$-dimensional subspace).

We believe that the $aBc$ problem is interesting in its own right as well as probably useful in other contexts, and understanding its communication complexity is important.
Also, in an  actual implementation, the final measurement against a 1-dimensional state should be easier to implement than against an $n/2$ dimensional subspace.

%\section{Related Work}

\section{Preliminaries}

\subsection{Manifolds}

We consider several compact Riemannian manifolds in this paper. All of them allow a uniform distribution via the Haar measure.

By $S^{n-1}$ we denote the sphere of real unit vectors in $\mathbb{R}^n$. $O_n$ denotes the set of real orthogonal $n\times n$-matrices.
The Stiefel manifolds $St_{k,n}$ consist of all $k$-tuples of unit vectors of dimension $n$ that form an orthogonal system.
We refer to \cite{ABoBa} for more information. $\sigma$ is used to denote the uniform distributions on a manifold, which is usually implicit.
We regard volume elements as normalized, so that on a compact manifold as considered here, a density function integrates to 1, and a uniform density is 1 everywhere.

\subsection{Differential R\'enyi Divergence}

In this paper all information theoretic notions are in the differential setting, i.e., for random variables that are not discrete.
We refer to \cite{div} for background on R\'enyi divergences.
\begin{defn}[Simple Orders] We call any $\alpha\in\mathbb{R}$ such that $\alpha\in(0, 1)\bigcup (1, \infty)$ a simple order.
\end{defn}

\begin{defn}[Extended Orders] $\alpha$ having the value 0, 1, or $\infty$ is called an extended order.
\end{defn}

\begin{defn}[R\'enyi Divergence]Let $P$ and $Q$ be two arbitrary distributions on a measurable space $(\Omega,\mathcal{F})$ that have density functions $p,q$. The R\'enyi divergence of (simple) order $\alpha$ of $P$ from $Q$ is defined as
\begin{equation}\label{D_alpha}
    D_\alpha(P||Q)=\frac{1}{\alpha -1}\ln\int p^\alpha q^{1-\alpha} d\sigma.\end{equation}
\noindent The R\'enyi divergence for extended orders are defined as follows:
\begin{equation}
    D_0(P||Q) = -\ln Q(p>0).
\end{equation}
\begin{equation}\label{D_1}
    D_1(P||Q) = D(P||Q) = \int p\ln\frac{p}{q} d\sigma.
\end{equation}
\begin{equation}
    D_\infty(P||Q) = \ln\Big(\text{ess}\sup_P\frac{p}{q}\Big).\end{equation}
\end{defn}

\begin{remark}
    $D_1(P||Q)$ is also known as the Kullback–Leibler divergence.
\end{remark}

The R\'enyi divergence for the discrete case is defined in a similar way, except that the integral is replaced by summation and densities are replaced with probabilities.
This leads to big problems when considering concepts like entropy, but works fine for divergences.

\begin{example}\label{eg1}
Let $S\subseteq S^{n-1}$ be a set that has measure $\sigma(S)$. Consider a density function $f: S^{n-1}\rightarrow\mathbb{R}^+$ that is for the distribution that is uniform on S, i.e.,
$$f(x)=\begin{cases} 0, \text{ if } x\notin S\\
\frac{1}{\sigma(S)}, \text{ if } x\in S\\
\end{cases}.$$
Then $D(f||unif) = (1-\frac{1}{\sigma(S)})\cdot 0+\sigma(S)\cdot\frac{1}{\sigma(S)}\ln(\frac{1}{\sigma(S)})=\ln\Big(\frac{1}{\sigma(S)}\Big),$ where $unif$ denotes the uniform distribution.
\end{example}

\begin{example}
Let $f:S^{n-1}\rightarrow\mathbb{R}^+$ be a density function and let $unif$ denote the uniform distribution. Then
$$D_2(f||unif) = \ln\int_{S^{n-1}} |f|^2 d\sigma$$
$$D_1(f||unif) = \int_{S^{n-1}} f\ln f d\sigma$$
\end{example}

\subsection{Spherical Harmonics}
Let $L^2(S^{n-1})$ denote the space of all square-integrable functions on $S^{n-1}$, i.e.,
$$f\in L^2(S^{n-1}) \iff\int_{S^{n-1}} |f|^2 d\sigma <\infty.$$
\begin{defn}[Spherical Harmonics ($S_k$) \cite{KR}]
    For any integer $k\geq 0$, $S_k$ denotes the spherical harmonics of degree $k$, which is the restriction to the sphere of all harmonic, homogeneous polynomials of degree $k$ in $\mathbb{R}^n$.
\end{defn}

Let $Proj_{S_k}$ denote the orthogonal projection operator onto $S_k$ Then for any $f\in L^2(S^{n-1})$, we have
$$f = \displaystyle\sum_{k=0}^\infty Proj_{S_k}f,$$ where the sum converges in $L^2(S^{n-1}).$

\subsection{Noise operator} The noise operator on $S^{n-1}$ is given by
$$U_\rho=\rho^{-\Delta},$$
where $\Delta$ is the spherical Laplacian\footnote{See \cite{KR} for more background.} for $0\leq\rho\leq 1$. Let $-\lambda_k=-k(k+n-2)$ denote the eigenvalues of the spherical Laplacian. Then, for any $k\geq 0$ and $\psi_k\in S_k$,
\begin{equation}U_\rho\psi_k=\rho^{\lambda_k}\psi_k.\end{equation}\label{noiseev}

\subsection{Hypercontractivity}
\begin{defn}[$p$-norm] Let $p\geq 1$ be a real number. The $p$-norm of a measurable function $f:S^{n-1}\rightarrow \mathbb{R}$ is given by
$$||f||_p=\Big(\int_{S^{n-1}}|f|^p d\sigma\Big)^{1/p}.$$
\end{defn}
\begin{defn}[$\infty$-norm] The $\infty$-norm of a measurable function $f:S^{n-1}\rightarrow \mathbb{R}$ is given by
$$||f||_\infty=\displaystyle\inf_{k\geq 0}\{k|\sigma(|f|>k)=0\},$$
where $\sigma$ is the uniform distribution.

\end{defn}

The hypercontractivity inequality on the sphere (again see \cite{KR} for more discussion about this) states that for any $1\leq p\leq q$ and any function $f\in L^p(S^{n-1})$,
\begin{equation}\label{hyper}
    ||U_\rho f||_q\leq ||f||_p,
\end{equation}
where $0\leq \rho \leq \Big(\frac{p-1}{q-1}\Big)^{1/(2n-2)}$ and $U_\rho$ is the noise operator mentioned in the previous subsection.

This means that $U_\rho$ is not merely a contraction, but that it contracts even when we increase the parameter of the norm.
In a given compact manifold hypercontractivity of the noise operator derived from the Laplace-Beltrami operator is equivalent to the truth of a log-Sobolev inequality, which can be deduced from the Bakry-Emery criterion (relying on a lower bound on the curvature of the manifold). Both of these also imply concentration of measure for Lipschitz functions. We recommend \cite{ABoBa} for an introduction to the subject.

\subsection{Radon Transform}
For any $y\in S^{n-1}$, we denote the uniform probability measure on the sphere $S^{n-1}\bigcap y^\perp$ as $\sigma_{y^\perp}$. Then the spherical Radon transform $R(f)$ of an integrable function $f: S^{n-1}\rightarrow \mathbb{R}$ is defined as
$$R(f)(y)=\int_{S^{n-1}\bigcap y^\perp}f(x)d\sigma_{y^\perp}(x).$$
Let $V_n$  be defined as follows:
\begin{equation}\label{V_n}
    V_n=\big\{(x, y)\in S^{n-1}\times S^{n-1}|x\cdot y=0\big\}.
\end{equation}
Note that $V_n$ is just the Stiefel manifold $St_{2,n}$.

We note the following observation \cite{KR}: for functions $f, g\in L^2(S^{n-1})$,
\begin{equation}\label{radon}
    \int_{V_n}f(x)g(y)d\sigma_V(x, y)=\int_{S^{n-1}}f(x)R(g(x))d\sigma.
\end{equation}
Define for all even $k\geq 0$,
$$\mu_k=(-1)^{k/2}\mathbb{E}[X_1^k]$$
where $X=(X_1,\cdots X_{n-1})$ is a random vector that is uniformly distributed in $S^{n-2}$. Then, $S_k$ is an eigenspace of $R$ with $\mu_k$ being the eigenvalue  \cite{KR}.
Odd $\mu_k$ are 0. $\mu_0=1$.

\subsection{Communication Complexity and Rectangles}

In the number-in-hand model Alice, Bob, Charlie receive inputs from $X\times Y\times Z$. Their task is to compute a (partial) Boolean function $f(x,y,z)$, or to compute an approximation to a real function $f(x,y,z)$. Each player has one input and knows that input only initially. All our protocols are either public coin randomized or distributional.

Our main lower bound applies to protocols of the following kind. The protocols are deterministic (or respectively distributional, after fixing random bits in the randomized case) and we consider a three-player protocol where Bob and Charlie can communicate with each other over many rounds and then send the resulting transcript to Alice, who has to produce the output (Bob$\leftrightarrow$Charlie$\rightarrow$Alice). The cost of such a protocol is the length of the transcript between Bob and Charlie.
Such protocols decompose into what we call one-way rectangles\footnote{This is a generalization of one-way rectangles in the two player case \cite{klauck:tradeoffs}.}
\begin{defn}[One-way Rectangle]
A three-player one-way rectangle $R$ is a set $S\times T$ such that $S\subseteq Y$, $T\subseteq Z$, and a function $Acc:X\mapsto\{0, 1\}$. Let $\mu$ be a distribution on $X\times Y\times Z$. The error of $R$ is given by $\frac{\mu\big((x,y,z)\in X\times R: Acc(x)\neq f(x,y,z)\big)}{\mu(R)}$ and the size $\mu(R)$ is $\mu(X\times S\times T)$.
\end{defn}

The corresponding lower bound for the distributional communication complexity of $f(x,y,z)$ is minus the logarithm of the size of the largest one-way rectangle with error $\epsilon$ under some distribution $\mu$. As usual, distributional complexity lower bounds randomized complexity \cite{kushilevitz&nisan:CC}.

We define $R^{C\to B\to A} (f)$ as the complexity of the cheapest one-way protocol for $f$, i.e., a protocol in which Charlie sends to Bob sends to Alice. $R^{C\leftrightarrow B\to A}(f)$ is the cost of the cheapest protocol as described above for $f$. If no probability of failure is indicated explicitly it is assumed to be 1/3.

\subsection{Pinsker's Inequality}\label{pinsker}
If $P$ and $Q$ are are two probability distributions on a measurable space $(\Omega,\mathcal{F})$, then
$$\delta(P, Q)\leq \sqrt{\frac{1}{2}D_1(P||Q),}$$
where $\delta(P, Q)=\displaystyle\sup_A\{|P(A)-Q(A)|\}$ is the total variational distance between $P$ and $Q$, where $A\in\mathcal{F}$ is a measurable event.

\subsection{Data Processing Inequality}
Let $P$ and $Q$ be probability measures defined on $(\Omega, \mathcal{F})$. Let $\mathcal{A}=\{A_i|i\geq 1\}$ be any partition of $\Omega$ and let $P_{\mathcal{A}}=\{P(A_i)|i\geq 1\}$ and $Q_{\mathcal{A}}=\{Q(A_i)|i\geq 1\}$. Then,
$$D_\alpha(P||Q)\geq D_\alpha(P_{\mathcal{A}}||Q_{\mathcal{A}}).$$
\subsection{Measurability}

We assume that all sets considered in this paper are Borel-measurable. In actual applications the inputs would be from a finite set defined by representing numbers with sufficient precision, say $1/poly(n)$. In this situation the problem $aBc$ needs to be re-defined to allow some slack, i.e., $aBc$ should be 1 if $a^TBc$ is very close to 1. Any algorithm for this relaxed problem yields an algorithm for the exact problem by rounding, and hence lower bounds for the exact problem yield lower bounds for the relaxed problem. Furthermore, in the relaxed problem rectangles are finite product sets, and rounding turns those into simple measurable sets.

\subsection{Constants}

We often use upper case letters for constants that are large enough and lower case letters for constants that are small enough. We usually do not track the value of constants very tightly, and so $C$ may mean something different in a proof, when moving from line to line.

\section{The Problem}

We consider variants of the problem to compute/approximate the product $a^TBc$, where $a,c$ are real unit vectors, and $B$ is an orthogonal matrix. This problem is considered both as a data-streaming problem, and as a communication complexity problem. In the data-streaming setting the vectors are streamed entry-wise, and the matrix row-by row entry-wise. More details about that later. The goal here is to approximate the product $a^TBc$ with additive error $\epsilon$, but we also consider the decision version, where it is promised that either $a^TBc=1$ or $a^TBc=-1$. This is in a sense the worst approximation still meaningful for this kind of product.

In communication complexity we consider the 3 player number-in-hand model. In the problem $aBc$, Alice is given a vector $a\in S^{n-1}$, Bob is given an $n\times n$ matrix $B$ from the orthogonal group $O_n$, and Charlie gets a vector $c\in S^{n-1}$. The problem is described by the following function:
$$aBc(a, B, c)=1\Leftrightarrow a^T\cdot B\cdot c=1,\hspace{10mm}aBc(a, B, c)=-1\Leftrightarrow a^T\cdot B\cdot c=-1.$$

We also consider the problem of approximating $a^TBc$ within additive error $\epsilon$.

\section{Notations}

The following notations are crucial to this paper.

\begin{center}
    \begin{tabular}{ |p{2cm}|p{3cm}|p{10cm}| }
     \hline
     \multicolumn{3}{|c|}{Notations, Mappings and Definitions} \\
     \hline
    Notation & Mapping & Definition\\
     \hline
    $f$  & $S^{n-1}\rightarrow \mathbb{R}^+$ & A density function of a distribution on the sphere \\
    $H$ & \hspace{1.5cm} - & A uniformly chosen hyperplane \\
    $f_{|H}$ & $S^{n-1}\bigcap H\rightarrow \mathbb{R}^+$ & Density function $f$ restricted to a randomly chosen equator \\
    $\bar{f}_{|H}$ & $S^{n-1}\bigcap H\rightarrow \mathbb{R}^+$ & Normalized version of $f_{|H}$\\
    ${unif}$ & $S^{n-1}\rightarrow\mathbb{R}^+$ & Uniform  distribution on the sphere\\
    ${unif}_{S^{n-1}\bigcap H}$ & $S^{n-1}\bigcap H\rightarrow\mathbb{R}^+$ & Uniform  distribution on the sphere on equator $H$\\
    $D_{\alpha}(f||g)$ & \hspace{1.5cm}- & The $\alpha$-R\'enyi divergence of $f$ from $g$\\
     \hline
    \end{tabular}
\end{center}

$\sigma$ is used for uniform densities on manifolds.

\section{Upper Bounds}
\subsection{Communication Complexity}
We first present an $O(\sqrt{n})$ communication protocol for the decision problem $aBc$ stated above. This protocols is also presented in our paper \cite{KLim}. We later show how to modify the protocol so as to work without the promise, with additive error $\epsilon$.
\begin{enumerate}
    \item Charlie and Bob share a set $T$ of $2^{O(k)}$ random unit vectors $w\in S^{n-1}$ as public coin, where $k$ is a parameter to be determined later. Among the $2^{O(k)}$ vectors shared with Bob, Charlie computes $W_{max}=argmax_{w\in T}\{\braket{w|c}\}$.
    \begin{lem}
       Define $T$ as a set of vectors randomly drawn from $S^{n-1}$ under the Haar measure (the unique rotationally-invariant probability measure on $S^{n-1}$ ), such that $|T|={32\sqrt{k}}e^{2k}$. If $v\in S^{n-1}$ is a fixed vector, then there exists a $w\in T$ that has an inner product with $v$ that is greater than $\sqrt{\frac{k}{n}}$ with high probability, for  all $1\leq k\leq \frac{n}{4}$.
    \end{lem}
    \begin{proof}
       According to Lemma 1 in \cite{caps}, Pr($\braket{v, w}^2\geq\frac{k}{n})\geq \frac{e^{-k}}{16\sqrt{k}}$ for $w\in S^{n-1}$ uniformly at random. We have  Pr($\braket{v, w}\geq\sqrt{\frac{k}{n}})\geq \frac{e^{-k}}{32\sqrt{k}}$ due to the fact that $\braket{v, w}$ could be negative. By the definition of $T$, we have that $$Pr(\forall w\in T:\braket{v, w}\leq\sqrt{\frac{k}{n}})\leq \big(1-\frac{e^{-k}}{32\sqrt{k}}\big)^{{32\sqrt{k}}e^{2k}}=\Big[\Big(1-\frac{1}{32\sqrt{k}e^{k}}\Big)^{32\sqrt{k}e^{k}}\Big]^{e^k}\leq({\frac{1}{e})}^{e^k}.$$ In other words, the probability of all $w$'s in the sample having an inner product with $v$ that is less than $\sqrt{\frac{k}{n}}$, is extremely small. This implies that there exists a $w\in T$ such that $\braket{v, w}\geq\sqrt{\frac{k}{n}}$ with high probability.
    \end{proof}

    Recall that $W_{max}$ is the vector that maximizes the inner product with $c$, then
    $$W_{max}=\alpha\ket{c}+\sqrt{1-\alpha^2}\ket{\sigma},$$
    where $\sigma\bot c$ and $\alpha\geq\sqrt{\frac{k}{n}}$.
    \item Next, Charlie sends the name of $W_{max}$ to Bob. This requires $O(k)$ communication. Bob then computes the following: $$B\ket{W_{max}}=\alpha B\ket{c}+\sqrt{1-\alpha^2}B\ket{\sigma}.$$
    \item Alice and Bob then jointly estimate the inner product between $B\ket{W_{max}}$ and $a$ by using the protocol proposed by Kremer, Nisan and Ron \cite{KNR}.
    \begin{fact}[Inner Product Estimation Protocol by Kremer, Nisan and Ron \cite{KNR}]
       The inner product estimation protocol approximates the inner product between two vectors from $S^{n-1}$ within $\epsilon$ additive error, which requires communication $O(\frac{1}{\epsilon^2})$.
    \end{fact}
    \begin{align*}
        \begin{split}
            \bra{a}B\ket{W_{max}}
            & =\alpha \bra{a}B\ket{c}+\sqrt{1-\alpha^2}\bra{a}B\ket{\sigma}\\
            & = \pm\alpha+\sqrt{1-\alpha^2}\bra{a}B\ket{\sigma},\\
        \end{split}
    \end{align*}
    where $\sqrt{1-\alpha^2}\bra{a}B\ket{\sigma}=0$ since $\sigma\bot c$ and $B^T a$ is either equal to $c$ or $-c$. That is to say,
    \[\bra{a}B\ket{W_{max}}
          \begin{cases}
            \geq \sqrt{\frac{k}{n}},\hspace{1mm} \text{for}\hspace{1mm} \text{1-inputs}\\
            \leq -\sqrt{\frac{k}{n}},\hspace{1mm} \text{for}\hspace{1mm}\text{{-1}-inputs}.\\
         \end{cases}
        \]
    Setting $\epsilon$ to be smaller than $\sqrt{\frac{k}{n}}$, say $\frac{1}{100}\sqrt{\frac{k}{n}}$ to allow for sufficient separation between {-1}- and 1-inputs, Kremer, Nisan and Ron's protocol requires $O(\frac{n}{k})$ communication.
    \end{enumerate}
    In order to minimize the total amount of communication ($O(k)$ in Step 2 and $O(\frac{n}{k})$ in Step 3), we set $k=\sqrt{n}$. Therefore, the total amount of communication required for the protocol equals $O(\sqrt{n})$.

We now describe the modifications necessary to allow us to approximate $a^TBc$ within additive error $\epsilon$.
In step 2 Charlie also sends the value of $\alpha=\langle c| W_{max}\rangle$, as a number with precision $1/poly(n)$, using $O(\log n)$ additional bits of communication.

In step 3 it is no longer true that $\sigma\perp B^Ta$, but $\sigma$ is a uniformly random vector from $S^{n-1}\cap c^\perp$.
This implies that $|\bra{a}B\ket{\sigma}|\leq 10/\sqrt n$ with probability at least $1-2e^{-50}$ by standard estimates on the area of spherical caps and the observation that $B^Ta=\bra{a}B\ket{c} c +\sqrt{1-\bra{a}B\ket{c}^2} \theta$ for some unit vector $\theta\perp c$. $\theta$ is a fixed vector orthogonal to $c$ and $\sigma$ is a random vector orthogonal to $c$.
Hence $\bra{a}B\ket{\sigma}=\sqrt{1-\bra{a}B\ket{c}^2}\langle \theta \ket{\sigma}$, and hence is smaller in absolute value than the inner product with a random vector, unless $B^Ta\perp c$, in which it is equally large.

The players may hence work in stage 3 as before, but with error $\frac{1}{100}\alpha\cdot\epsilon$. The obtained result is
scaled by multiplying by $1/\alpha$ and used as the output.

\begin{theorem}
For every $\epsilon>0$ there is a 3 player number-in-hand one-way protocol that approximates the value of $a^TBc$ with additive error $\epsilon$ and communication $O(\sqrt n/\epsilon^2)$.

Furthermore $R^{C\to B\to A}(aBc)=O(\sqrt n)$.
\end{theorem}

We also note that the above protocol for $aBc$ can be used so that Charlie sends $k$ bits of communication, and Bob $O(n/k)$ bits for any $k$ larger than some constant.
In particular there is a one-way protocol in which Charlie sends $O(n^{2/3})$ bits and Bob $O(n^{1/3})$ bits. In our lower bound we will show that this is optimal, i.e., that with less communication from one of the players the error must be large. In that sense our lower bound is tight, but only for a certain tradeoff.

\subsection{Data-streaming}

We now consider data-streaming algorithms for approximating $a^TBc$ efficiently within additive error $\epsilon$, and for deciding the $aBc$ problem. It is trivial to compute $a^TBc$ exactly using $O(n\log n)$ space, where we assume that first on the data-stream $c$ is presented entry-wise with precision $1/poly(n)$, then $B$ row by row entry-wise, then $a$.

 A first idea would be to follow the approach of Alon, Matias, and Szegedy \cite{ams:frequency} in their algorithm for moment estimation, which readily adapts to the estimation of inner products. Their algorithm  is randomized and computes the inner product $a^Tc$ {\bf in expectation}. Luckily, the variance can also be bounded, and as a result with enough parallel repetitions one can approximate the inner product to within additive error $\epsilon$ with space $O(\log n/\epsilon^2)$.

 Their algorithm uses the following estimator. Let $V=\{v_1,\ldots, v_h\}$ denote a set of $h=O(n^2)$ four-wise independent random variables, where each $v_j(i)$ is 1 or -1.
 A construction of such a family is presented in \cite{four-wise} based on the parity check matrices of BCH codes.
 One first chooses a random $v$ from $V$, and then on input vectors $a,c$ computes $\sum v(i) a_i$ as well as $\sum v(i) c_i$, and multiplies those. Using $O(1/\epsilon^2)$ of these estimators and averaging appropriately leads to a good approximation of $a^Tc$.

Attempting to generalize this algorithm to the $a^TBc$ scenario, one can consider using the following estimator:
\[\left(\sum_i v(i)a_i\right)\left(\sum_{i,j} v(i)w(j) B_{i,j}\right)\left(\sum_j w(j) c_j\right),\]
where $v,w$ are random vectors from an 8-wise independent family.

Using this one still obtains an algorithm that computes $a^TBc$ in expectation, but the variance can be shown to be $n+O(1)$, making the result completely unreliable unless one is willing to make $O(n)$ parallel repetitions, using space $O(n\log n)$ and hence not outperforming the trivial algorithm.

However, something better can be achieved by trying to emulate our communication protocol and combining the first step of the protocol with the idea for computing simple inner products. The problem with the communication protocol is that we cannot store random vectors from the sphere efficiently. Step 1 of the protocol can be considered as putting a net on the sphere
and finding the closest vector from the net. We try something similar.

We describe an efficient streaming algorithm for the decision version $aBc$.
While processing $c$ the algorithm maintains a set of $100\sqrt n$ positions $i$ with largest $|c_i|$, as well as the corresponding $c_i$. This can be achieved easily by keeping them sorted, removing the smallest positions when necessary.
 After $c$ has streamed we normalize the resulting vector $\tilde{c}$ of at most $100\sqrt n$ non-zero positions.
 The normalized vector $\bar{c}$ is our ``approximation'' of $c$.

The key observation is that $\langle c\ket{\tilde{c}}\geq 100/\sqrt n$ and that $c-\tilde{c}\perp\tilde{c}$.
Then $\langle c\ket{\bar{c}}=\alpha$ with $\alpha\geq \sqrt{100/\sqrt n}=10n^{-1/4}$, because $||\tilde{c}||_2^2=\langle c\ket{\tilde{c}}$.
Hence $\bar{c}=\alpha c+\sqrt{1-\alpha^2}\hat{c}$, where $\hat{c}\perp c$.

Our goal is then to approximate the inner product $\bra{a}B\ket{\bar{c}}$ within some error.
Note that \[\bra{a}B\ket{\bar{c}}=a^TB (\alpha c+\sqrt{1-\alpha^2}\hat{c}),\]
hence \[\bra{a}B\ket{\bar{c}}=\alpha a^TBc + \sqrt{1-\alpha^2}a^TB \hat{c}).\]
But $B^Ta$ is either $c$ or $-c$, hence orthogonal to $\hat{c}$, and hence the second term vanishes.
So $\bra{a}B\ket{\bar{c}}=\alpha a^TBc$, for a known $\alpha\geq 10n^{-1/4}$.

We may now use the estimator \[\left(\sum_i v(i) a_i\right)\left(\sum_i v(i)(B\bar{c})_i\right),\]
where each $(B\bar{c})_i$ can be computed from the stored $\bar{c}$ while row $i$ of $B$ is streamed.

The error analysis is exactly as in \cite{ams:frequency} and we need to use additive error at most $n^{-1/4}$ in order to see whether $a^TBc$ is 1 or -1.
This leads to an algorithm using space $O(\sqrt n\log n)$.

\begin{theorem}
There is a data-streaming algorithm for the (decision) problem $aBc$, which succeeds with high probability and uses space $O(\sqrt n\log n)$.
\end{theorem}

We conjecture that an algorithm exists that approximates $a^TBc$ with additive error $\epsilon$ and complexity $O(\sqrt n\log n/\epsilon^2)$.
The main task there is to find efficiently computable nets on the sphere. Note that in our regime of parameters these nets are small and closeness needs to be defined via inner products.

\section{Equator Sampling Theorems}
In this section, we show that a density function $f$ of a distribution when restricted to a random equator (and normalized) is as close to the uniform distribution as $f$ is on the whole sphere with high probability. The natural approach is to use the concept of R\'enyi divergence, which characterizes how "close"  probability distributions. In particular, we show that with high probability,
$$|D_\alpha(\bar{f}_{|H}||{unif}_{S^{n-1}\bigcap H})-D_\alpha(f||unif)|\leq t,$$
for $0<t<1$ and $\alpha>1$. We also show a slightly different result for $\alpha=1$ and another result about the divergence of $f$ from $g$ instead of  $f$ from $unif$.

We draw on \cite{KR} but modify their notions and results to suit our context of R\'enyi divergence. The main modification to their approach unfortunately is right in the core of their proof, so we need to reproduce two of their technical results with the necessary modifications. We keep notations similar to ease comparison.

The following lemma is an improved version of Lemma 5.3 in \cite{KR}. We provide an upper bound on the projection length of a density $f$ onto the space $S_k$ in terms of $D_2(f||unif)$. Our main modification here is to prove an upper bound on the $p$-norm (for $1\leq p\leq 2$) of a function in terms of the 2-norm.
\cite{KR} provide such a bound in terms of the $\infty$-norm only, which will not work in our case\footnote{The $\infty$-norm of a density function may be large in general and then their result can't be applied.}.

\begin{lem}\label{5.3}
    Considering a density function $f:S^{n-1}\rightarrow\mathbb{R}$ for any $k>1$ we have
    $$||Proj_{S_k}f||_2\leq\Bigg(e\cdot \max\Bigg(1, \frac{D_2(f||unif)}{\lambda_k/(2n-2)}\Bigg)\Bigg)^{\frac{\lambda_k}{2n- 2}},$$
    where the $-\lambda_k$ are eigenvalues of the spherical Laplacian.
\end{lem}
\begin{proof}
    First, note that for any $p=1+\epsilon\leq 2$,
    \begin{align*}
        ||f||_p & = ||f||_{1+\epsilon}\\
        & = \Bigg(\int_{S^{n-1}}|f|^{1+\epsilon} dx\Bigg)^{1/(1+\epsilon)}\\
        & = \Big(\mathbb{E}[|f|^\epsilon]\Big)^{1/(1+\epsilon)}\\
        & \leq \Big(\mathbb{E}[|f|]\Big)^{\epsilon/(1+\epsilon)}\\
        & = \Big(\int_{S^{n-1}}|f|^2 dx\Big)^{\epsilon/(1+\epsilon)}\\
        & = \Big(||f||_2^2\Big)^{\epsilon/(1+\epsilon)}\\
        & = \Big(||f||_2^2\Big)^{(p-1)/p}\\
        & \leq \Big(||f||_2^2\Big)^{p-1},\\
    \end{align*}
\noindent which is equivalent to
$$D_p(f||unif)\leq D_2(f||unif).$$ In other words, R\'enyi divergences do not decrease when the parameter $p$ is increased. This inequality holds for all pairs of distributions, i.e. one can replace the uniform distribution with any other  distribution  \cite{div}. Furthermore, 2 can be replaced with any number larger than $p$.

Since $||Proj_{S_k}f||_2\leq||f||_2\leq e^{D_2(f||unif)/2}$, the lemma holds for $k$ with $\lambda_k>(n-1)D_2(f||unif)$. So we can assume that $k$ is such that $\lambda_k\leq (n-1) D_2(f||unif)$ from now on. We set $q=2$ using (\ref{hyper}) and get that for any $1\leq p\leq 2$,
\begin{equation*}
||U_\rho f||_2\leq ||f||_p\leq \Big(||f||_2^2\Big)^{p-1},\end{equation*} where $\rho=(p-1)^{1/(2n-2)}$.

When one projects $f$ to the space $S_k$, the result is that for every $1\leq p\leq 2$,
\begin{equation}\label{hyp}
(p-1)^{\frac{\lambda_k}{2n-2}}||Proj_{S_k}f||_2=||Proj_{S_k}(U_\rho f)||_2\leq ||U_\rho f||_2\leq \Big(||f||_2^2\Big)^{p-1}\end{equation}

We may then choose $p=1+\frac{\lambda_k}{(2n-2)D_2(f||unif)}\leq 1.5<2$, which establishes the stated bound.
Note here that the scalar on the left-hand-side of (\ref{hyp}) is from the scaling the noise operator does on $S_k$.
\end{proof}

\begin{remark}
In the proof of Lemma 5.3 in \cite{KR}, the authors use $\ln||f||_\infty$ instead of $D_2(f||unif)$. This is sufficient for flat density functions which are uniform on a subset. However in our case of arbitrary density functions, we cannot put a useful bound on the infinity norm.
\end{remark}

In order to help us prove the concentration bound between $D_\alpha(\bar{f}_{|H}||{unif}_{S^{n-1}\bigcap H})$ and $D_\alpha(f||unif)$, we require the following (modified) theorem from \cite{KR}.
Again the change is to replace $\ln||\cdot||_\infty$ by the 2-R\'enyi divergence.

\begin{theorem}\label{5.2}
Suppose $f, g:S^{n-1}\rightarrow \mathbb{R}^+$ are density functions, and let
$$s=(D_2(f||unif)+1)(D_2(g||unif)+1).$$
Then when $s\leq cn$,
$$\Bigg|\int_{V_n} f(x)g(y)d\sigma V(x, y)-1\Bigg|\leq\frac{Cs}{n},$$
where $C, c>0$ are constants.
\end{theorem}
\begin{proofs}
The proof is similar to that of Theorem 5.2 of \cite{KR} but with some alterations.

We recall that $\mu_{2k}=(-1)^k\mathbb{E}[X^{2k}]$, where $X$ is the first entry in an $n-1$ dimensional vector uniformly at random from $S^{n-2}$ with $\mathbb{R}^{n-1}$ being the ambient space. All odd $\mu_{2k+1}$ are set to 0 and $\mu_0=1$.
Lemma 5.4 of \cite{KR} shows that $\mu_k$ is the eigenvalue for the Radon transform on $S_k$, which is an eigenspace. In their Lemma 5.5 $|\mu_k|$  is upper bounded by $(C\frac kn)^{k/2}$ for a constant $C$ when $k\geq 2 $ and $n$ large enough.

By (\ref{radon}) and the fact that $S_k$ is an eigenspace, we have
$$\int_{V_n}f(x)g(y)d\sigma_V(x, y)=\int_{S^{n-1}}fR(g)d\sigma=\displaystyle\sum_{k=0}^\infty\mu_k\int_{S^{n-1}}Proj_{S_k}(f)Proj_{S_k}(g)d\sigma.$$
By the Cauchy–Schwarz inequality and the fact that $\mu_0=1$ and that $Proj_{S_0} (f)$ is the function that is constant 1 (same for $g$) and that $\mu_k=0$ for odd $k$
$$\Bigg|\int_{V_n}f(x)g(y)d\sigma_V(x, y)-1\Bigg|\leq \displaystyle\sum_{k=1}^{\infty}|\mu_{2k}|||Proj_{S_{2k}}f||_2||Proj_{S_{2k}}g||_2.$$
In order to prove the theorem, we have to  show that the sum on the right hand side of the above inequality is at most $C\gamma\beta/n$, where $\gamma=D_2(f||unif)+1$ and $\beta=D_2(g||unif)+1$. Note hat $1\leq\gamma,\beta$ and by the assumption of the theorem $\beta\gamma\leq cn$. We first have to bound the part of the sum where $k$ runs from 1 to T-1, where $T=\floor{\delta n}$, for a sufficiently small positive constant $\delta$. Using our Lemma \ref{5.3} and Lemma 5.5 of \cite{KR}, we have the following upper bounds:
$$|\mu_{2k}|\leq (\frac{Ck}{n})^k,$$
$$||Proj_{S_{2k}} f||_2\leq \Big(C\cdot \max\Big\{1, \frac{\gamma}{k}\Big\}\Big)^{\frac{\lambda_{2k}}{(2n-2)}},$$
and similarly for $g$ and $\beta$. Hence,
$$\displaystyle\sum_{k=1}^{T-1}|\mu_{2k}|||Proj_{S_{2k}}f||_2||Proj_{S_{2k}}g||_2\leq \displaystyle\sum_{k=1}^{T-1}\big(\frac{Ck}{n}\big)^k\Big(C\cdot\max\Big\{1,\frac{\gamma}{k}\Big\}\Big)^{\frac{\lambda_{2k}}{(2n-2)}}
\Big(C\cdot\max\Big\{1,\frac{\beta}{k}\Big\}\Big)^{\frac{\lambda_{2k}}{(2n-2)}}$$

Considering the $k=1$ term we get that it is at most $C\beta\gamma/n$. Here we use that $\beta,\gamma\geq 1$. The remaining proof is the same as in \cite{KR}: For the sum up to $T-1$ one can show that the terms decay geometrically.
Furthermore, the sum from $T$ to $\infty$ can be treated as in \cite{KR} as well.

\end{proofs}

We now prove the first of our  main results about equator sampling.
In general we use the above techniques similarly to \cite{KR} but with our improved Theorem \ref{5.2}. The main idea is to apply it to various normalized nonnegative functions instead of density functions that are uniform on a subset. This allows us to show concentration results for a much richer class of functions.
Another difference is that we usually need to check a number of criteria instead of just one.

\begin{theorem}\label{5.1}
For any $0<t<1$ and $\alpha>1$,
$$\Pr\big[\big|D_\alpha(\bar{f}_{|H}||unif_{S^{n-1}\bigcap H})-D_\alpha(f||unif)\big|\geq t\big]\leq Be^{-bnt(\alpha-1)/\alpha  (D_{2\alpha}(f||unif)+1)},$$
for some constants $B,b>0$ independent of $\alpha$.
\end{theorem}

Note that the probability statement includes the case where $\bar{f}_{|H}$ does not exist, because $f$ integrates to 0 on the sphere in $H$.

\begin{proof}
We prove the following claims:

\begin{claim}\label{c1}
With high probability, little re-normalization is required to make $f_{|H}$ a density function.
\end{claim}
\begin{proof}
Let $E$ be the set of all $y\in S^{n-1}$ for which the hyperplane $H\subset\mathbb{R}^n$ orthogonal to $y$ satisfies
$$\int_{S^{n-1}\bigcap H}f_{|H}dx\geq 1+t,$$
where $0<t<1$. Let $g$ be uniform on $E$, i.e.
$$g(y)=\begin{cases}1/\sigma(E),\text{ on $E$}\\
0,\text{elsewhere}\\
\end{cases}$$
Then,
$$\int_{V_n}f(x)g(y)d\sigma_V(x, y)\geq 1+t,$$
for $V_n$ defined in (\ref{V_n}).
By Theorem \ref{5.2},
$$t\leq \frac{C\cdot (D_2(f||unif)+1)(D_2(g||unif)+1)}{n}.$$
We use here that $(D_2(f||unif)+1)(D_x(g||unif)+1)\leq cn$, because otherwise the bound becomes $t\leq 1$, which is true by the condition of the theorem. We will not mention this technicality in further iterations of this argument.
After rearranging, we get
$$D_2(g||unif)\geq \frac{tn}{C\cdot (D_2(f||unif)+1)}-1.$$
\begin{fact}\label{f1}
If $g$ is a density function uniform on a subset, then we have $D_1(g||unif)=D_2(g||unif)=\cdots=D_\alpha(g||unif)$ for all $\alpha\geq 1$.
\end{fact}

From Fact \ref{f1} and Example \ref{eg1}, we have $D_2(g||unif)=\ln\big(\frac{1}{\sigma(E)}\big)$ and hence
$$\frac{1}{\sigma(E)}\geq e^{\frac{tn}{C\cdot (D_2(f||unif)+1)}-1}$$
$$\sigma(E)\leq e^{-\frac{tn}{C\cdot (D_2(f||unif)+1) }}.$$
Repeating a similar argument for the lower bound on $\sigma(E)$ we obtain the following
\begin{equation}\label{M}
    \Pr\Bigg[\Bigg|\int_{S^{n-1}\bigcap H} f_{|H}dx-1\Bigg|\geq t\Bigg]\leq e^{-\frac{tn}{C\cdot (D_2(f||unif)+1)}},
\end{equation}
which means that with high probability, $f_{|H}$ is close to being a density function.
\end{proof}

\begin{claim}\label{c2}
With high probability, $f_{|H}$ is as close to the uniform distribution as $f$.
\end{claim}
\begin{proof}
Let $h=\frac{f^\alpha}{\int_{S^{n-1}} f^\alpha dz}$, where $||h||_1=1$ and let $E$ be the set of all $y\in S^{n-1}$ for which the hyperplane $H\subset \mathbb{R}^n$ orthogonal to $y$ satisfies
$$\int_{S^{n-1}\bigcap H}h\hspace{1mm}dx\geq 1+t,$$
where $0<t<1$. By Theorem \ref{5.2},
$$t\leq \frac{C\cdot (D_2(h||unif)+1)(D_2(g||unif)+1)}{n}.$$
After rearranging, we get
$$D_2(g||unif)\geq \frac{tn}{C\cdot (D_2(h||unif)+1)}-1.$$
Similar to the proof of Claim \ref{c1}, we get
$$\sigma(E)\leq e^{-\frac{tn}{C\cdot (D_2(h||unif)+1)}}.$$
Next, we would like to express the above inequality in terms of $D_{2\alpha}(f||unif)$. Notice that
\begin{align*}
    D_2(h||unif) & =\ln\int_{S^{n-1}\bigcap H} h^2dx=\ln\int_{S^{n-1}\bigcap H}  \Bigg(\frac{f^{2\alpha}}{\big(\int_{S^{n-1}} f^\alpha dz\big)^2}\Bigg)dx\\
    & = \ln\int_{S^{n-1}\bigcap H}f^{2\alpha}dx-2\ln\int_{S^{n-1}}f^\alpha dz\\
    & \leq \ln\int_{S^{n-1}\bigcap H}f^{2\alpha}dx\\
    & = (2\alpha -1)D_{2\alpha}(f||unif).\\
\end{align*}
Therefore,
$$\sigma(E)\leq e^{-\frac{tn}{C\cdot ((2\alpha -1)D_{2\alpha}(f||unif)+1) }}      ,$$
and hence
\begin{align*}
    & \Pr\Bigg[D_\alpha(f_{|H}||{unif}_{S^{n-1}\bigcap H}) - D_\alpha(f||{unif}_{S^{n-1}})\geq \frac{t}{\alpha-1}\Bigg]\\
    & = \Pr\Bigg[\ln\int_{S^{n-1}\bigcap H}f_{|H}^\alpha dx - \ln\int_{S^{n-1}}f^\alpha\geq t\Bigg]\\
    & = \Pr\Bigg[\ln\int_{S^{n-1}\bigcap H}\frac{f^\alpha}{\int_{S^{n-1}}f^\alpha dz}dx\geq t\Bigg]\\
    & \leq \Pr\Bigg[\int_{S^{n-1}\bigcap H}\frac{f^\alpha}{\int_{S^{n-1}}f^\alpha dz}dx\geq 1+ t\Bigg]\\
    & = \Pr\Bigg[\int_{S^{n-1}\bigcap H} h dx \geq 1+t\Bigg]\\
    & \leq e^{-\frac{tn}{C((2\alpha-1)\cdot D_{2\alpha}(f||unif)+1)}}.
\end{align*}
Repeating a similar argument for the lower bound on $\sigma(E)$, we get
$$\Pr\Bigg[\Bigg|D_\alpha(f_{|H}||{unif}_{S^{n-1}\bigcap H}) - D_\alpha(f||{unif}_{S^{n-1}})\Bigg|\geq \frac{t}{\alpha-1}\Bigg]\leq e^{-\frac{tn}{C((2\alpha-1)\cdot D_{2\alpha}(f||unif)+1)}}.$$
By replacing $\frac{t}{\alpha-1}$ with $t$, we get
\begin{equation}\label{N}
    \Pr\Bigg[\Bigg|D_\alpha(f_{|H}||{unif}_{S^{n-1}\bigcap H}) - D_\alpha(f||{unif}_{S^{n-1}})\Bigg|\geq t\Bigg]\leq e^{-\frac{tn(\alpha-1)}{C((2\alpha-1)\cdot D_{2\alpha}(f||unif)+1)}},
\end{equation}
which means with high probability, $f_{|H}$ is as close to the uniform distribution as $f$.
\end{proof}

Next, we show that the R\'enyi divergence of $f_{|H}$ from the uniform distribution changes by at most $\frac{2\alpha}{\alpha-1}\Big|\int_{S^{n-1}\bigcap H} f_{|H}dx-1\Big|$ after normalizing.
\begin{claim}\label{c3}
Assume that $\int_{S^{n-1}\bigcap H}f_{|H} dx=1-\ell$, where $\ell\leq t\leq\frac{1}{2}$, for the $t$ from Claim \ref{c1}. Then, $\bar{f}_{|H}=\frac{1}{1-\ell}f_{|H}$. We have the following bound:
$$D_\alpha(\bar{f}_{|H}||{unif}_{S^{n-1}\bigcap H})-D_\alpha(f_{|H}||{unif}_{S^{n-1}\bigcap H})|\leq\frac{2\alpha}{\alpha-1}\Bigg|\int_{S^{n-1}\bigcap H}f_{|H}dx-1\Bigg|.$$
\end{claim}
\begin{proof}
Observe that
\begin{align*}
    D_\alpha(\bar{f}_{|H}||{unif}_{S^{n-1}\bigcap H}) & = D_\alpha\Big(\frac{1}{1-\ell}f_{|H}\Big|\Big|{unif}_{S^{n-1}\bigcap H}\Big)\\
    & = \frac{1}{\alpha-1}\ln\int_{S^{n-1}\bigcap H}f_{|H}^\alpha\cdot\Big(\frac{1}{\int_{S^{n-1}\bigcap H}f_{|H}dx}\Big)^\alpha dx\\
    & = \frac{1}{\alpha-1}\Big[\alpha\ln\Big(\frac{1}{\int_{S^{n-1}\bigcap H}f_{|H}dx}\Big)+\ln\int_{S^{n-1}\bigcap H}f_{|H}^\alpha dx\Big]\\
    & = \frac{1}{\alpha-1}\Big[-\alpha\ln\Big(\int_{S^{n-1}\bigcap H}f_{|H}dx\Big)+\ln\int_{S^{n-1}\bigcap H}f_{|H}^\alpha dx\Big]\\
    & = \frac{1}{\alpha-1}\Big[-\alpha\ln(1-\ell)+\ln\int_{S^{n-1}\bigcap H}f_{|H}^\alpha dx\Big]\\
    & \leq \frac{2\alpha\ell}{\alpha-1}+D_\alpha(f_{|H}||{unif}_{S^{n-1}\bigcap H})\\
    & = \frac{2\alpha}{\alpha -1}\Big|\int_{S^{n-1}\bigcap H} f_{|H} dx-1\Big|+D_\alpha(f_{|H}||{unif}_{S^{n-1}\bigcap H}).\\
\end{align*}
Applying a similar argument for the lower bound on $D_\alpha(\bar{f}_{|H}||{unif}_{S^{n-1}\bigcap H})$, we prove our claim.
\end{proof}

Now we show that $\bar{f}_{|H}$ is still as close to the uniform distribution as $f$. By combining Claims \ref{c1}, \ref{c2}, \ref{c3} for all $0<t<1$,
\begin{align*}
    &\Pr\Big[\big|D_\alpha(\bar{f}_{|H}||{unif}_{S^{n-1}\bigcap H})-D_\alpha(f||unif)\big|\geq t\Big]\\
    & \leq \Pr\Big[\big|D_\alpha(\bar{f}_{|H}||{unif}_{S^{n-1}\bigcap H})-D_\alpha(f_{|H}||{unif}_{S^{n-1}\bigcap H})\Big|+\Big|D_\alpha(f_{|H}||{unif}_{S^{n-1}\bigcap H})-D_\alpha(f||unif)\Big|\geq t\Big]\\
    & \leq \Pr\Big[\big|D_\alpha(\bar{f}_{|H}||{unif}_{S^{n-1}\bigcap H})-D_\alpha(f_{|H}||{unif}_{S^{n-1}\bigcap H})\Big|
    \geq t/2\Big]\\
    & \hspace{5mm} +\Pr\Big[\big|D_\alpha(f_{|H}||{unif}_{S^{n-1}\bigcap H})-D_\alpha(f||unif)\Big|\geq t/2\Big]\\
    & \leq \Pr\Bigg[\bigg|\int_{S^{n-1}\bigcap H}f_{|H}\hspace{1mm}dx -1\Bigg|\geq \frac{(\alpha-1)t}{4\alpha}\Bigg]+\Pr\Bigg[\bigg|D_\alpha(f_{|H}||{unif}_{S^{n-1}\bigcap H})-D_\alpha(f||unif)\Bigg|\geq t/2\Bigg]\\
    & \leq e^{-\frac{(\alpha-1)nt}{\alpha C(D_2(f||unif)+1)}}+e^{-\frac{(\alpha-1)nt}{\alpha C(D_{2\alpha}(f||unif)+1)}}+\Pr\Big[\Big|\int_{S^{n-1}\bigcap H}f_{|H}dx-1\Big|\leq\frac{1}{2}\Big]\\
    & \leq Be^{-bnt(\alpha-1)/\alpha (D_{2\alpha}(f||unif)+1)}.
\end{align*}
\end{proof}

Note that we can use a limiting argument to get the same result for $D_\infty$, while such an approach fails for $\alpha=1$, arguably the most important case.
Now we consider the R\'enyi divergence for the case where $\alpha=1$.

\begin{theorem}\label{D1}
For any $0<t<1$,
$$\Pr\Bigg[\Bigg|\frac{D_1(\bar{f}_{|H}||{unif}_{S^{n-1}\bigcap H})}{D_1(f||unif)}-1\Bigg|\geq t\Bigg]\leq Ke^{-knt/(D_4(f||unif)+1)},$$
for some constants $K,k>0$.
\end{theorem}
\begin{proof}
We first prove the following claim:
\begin{claim}\label{c4}
With high probability, $f_{|H}$ is as close to the uniform distribution as $f$.
\end{claim}
\begin{proof}
Let $h=\frac{f\ln f}{\int_{S^{n-1}} f\ln f dz}$, where $||h||_1=1$ and let $E$ be the set of all $y\in S^{n-1}$ for which the hyperplane $H\subset \mathbb{R}^n$ orthogonal to $y$ satisfies
$$\int_{S^{n-1}\bigcap H}h\hspace{1mm}dx\geq 1+t,$$
where $0<t<1$. By Theorem \ref{5.2},
$$t\leq \frac{C\cdot (D_2(h||unif)+1)(D_2(g||unif)+1)}{n}.$$
After rearranging, we get
$$D_2(g||unif)\geq \frac{tn}{C(\cdot D_2(h||unif)+1)}-1.$$
Similar to the proof of Claim \ref{c1}, we get
$$\sigma(E)\leq e^{-\frac{tn}{C\cdot (D_2(h||unif)+1)}}.$$
Next, we would like to express the above inequality in terms of $D_4(f||unif)$. Notice that
\begin{align*}
    D_2(h||unif) & =\ln\int_{S^{n-1}\bigcap H} h^2dx=\ln\int_{S^{n-1}\bigcap H}\Bigg(\frac{(f\ln f)^2}{\big(\int_{S^{n-1}} f\ln f dz\big)^2}\Bigg)dx\\
    & = \ln\int_{S^{n-1}\bigcap H}(f\ln f)^2dx-2\ln\int_{S^{n-1}}(f \ln f) dz\\
    & \leq \ln\int_{S^{n-1}\bigcap H}(f\ln f)^2dx\\
    & \leq \ln\int_{S^{n-1}\bigcap H} f^4 dx\\
    & = 3D_4(f||unif)\\
\end{align*}
Therefore,
$$\sigma(E)\leq e^{-\frac{tn}{C\cdot (D_4(f||unif)+1)}}.$$ Recall that $E$ is the set of vectors that have the following property:
$$\int_{S^{n-1}\bigcap H}h\hspace{1mm} dx\geq 1+t.$$
Expressing $h$ in terms of $f$ and by definition of the KL-divergence, we get
\begin{gather*}
    \int_{S^{n-1}\bigcap H}\frac{f\ln f}{\int_{S^{n-1}}f\ln f  f dz}dx\geq 1+t\\
    \Leftrightarrow \frac{\int_{S^{n-1}\bigcap H}f\ln f dx}{\int_{S^{n-1}}f\ln f dz}\geq 1+t\\
    \Leftrightarrow \frac{\int_{S^{n-1}\bigcap H}f\ln f dx}{\int_{S^{n-1}}f\ln f dz}-1\geq t\\
    \Leftrightarrow \frac{D_1(f_{|H}||{unif}_{S^{n-1}\bigcap H})}{D_1(f||unif)}-1\geq t
\end{gather*}
Repeating a similar argument for the lower bound on $\sigma(E)$, we get
$$\Pr\Bigg[\Bigg|\frac{D_1(f_{|H}||{unif}_{S^{n-1}\bigcap H})}{D_1(f||unif)}-1\Bigg|\geq t\Bigg]\leq e^{-\frac{tn}{C\cdot( D_4(f||unif)+1)}}$$
which means with high probability, $f_{|H}$ is as close to the uniform distribution as $f$.
\end{proof}

Assume that $\int_{S^{n-1}\bigcap H}f_{|H} dx=1-\ell$, where $\ell\leq t\leq \frac{1}{2}$ for the $t$ in Claim \ref{c1}. Then $\bar{f}_{|H}=\frac{1}{1-\ell}f_{|H}$. Note that
\begin{align*}
    & \int_{S^{n-1}\bigcap H}\bar{f}_{|H}\ln\bar{f}_{|H}dx\\
    & = \int_{S^{n-1}\bigcap H}\frac{1}{1-\ell}f_{|H}\ln\Big(\frac{1}{1-\ell}f_{|H}\Big)dx\\
    & = \frac{1}{1-\ell}\int_{S^{n-1}\bigcap H}f_{|H}\ln f_{|H}dx+\frac{1}{1-\ell}\int_{S^{n-1}\bigcap H}f_{|H}\ln\Big(\frac{1}{1-\ell}\Big)dx\\
    & = \frac{1}{1-\ell}\int_{S^{n-1}\bigcap H}f_{|H}\ln f_{|H}dx+\frac{1}{1-\ell}\cdot\ln\Big(\frac{1}{1-\ell}\Big)\int_{S^{n-1}\bigcap H}f_{|H}dx\\
    & = \frac{1}{1-\ell}\int_{S^{n-1}\bigcap H}f_{|H}\ln f_{|H}dx+\frac{1}{1-\ell}\cdot\ln\Big(\frac{1}{1-\ell}\Big)\cdot (1-\ell)\\
    & = \frac{1}{1-\ell}\int_{S^{n-1}\bigcap H}f_{|H}\ln f_{|H}dx+\ln\Big(\frac{1}{1-\ell}\Big)\\
    & = \frac{1}{1-\ell}\cdot D_1(f_H||unif)-\ln{(1-\ell)}.\\
\end{align*}
Then,
\begin{align*}
     & \Pr\Bigg[\frac{D_1(\bar{f}_H||unif)}{D_1(f||unif)}\geq 1+t\Bigg]\\
    & = \Pr[D_1(\bar{f}_H||unif)\geq (1+t)\cdot D_1(f||unif)]\\
    & = \Pr\Bigg[\frac{1}{1-\ell}\cdot D_1(f_H||unif)-\ln{(1-\ell)}\Big)\geq (1+t)D_1(f||unif)\Bigg]\\
    & = \Pr\Bigg[\frac{D_1(f_h||unif)}{(1-\ell)(1+t)\cdot D_1(f||unif)}\geq 1+\frac{\ln{(1-\ell)}}{(1+t)\cdot D_1(f||unif)}\Bigg]\\
    & = \Pr\Bigg[\frac{D_1(f_H||unif)}{D_1(f||unif)}\geq (1-\ell)(1+t)+\frac{(1-\ell)\ln{(1-\ell)}}{D_1(f||unif)}\Bigg]\\
    & = \Pr\Bigg[\frac{D_1(f_H||unif)}{D_1(f||unif)}\geq 1+t-\ell-t\ell+\frac{(1-\ell)\ln{(1-\ell)}}{D_1(f||unif)}\Bigg].\\
\end{align*}
Bound $D_1(f||unif)$ below by a constant and set $t^\prime=t-\ell-t\ell+\frac{(1-\ell)\ln{(1-\ell)}}{D_1(f||unif)}$. Then,
\begin{align*}
    & \Pr\Bigg[\frac{D_1(f_H||unif)}{D_1(f||unif)}\geq 1+t-\ell-t\ell+\frac{(1-\ell)\ln{(1-\ell)}}{D_1(f||unif)}\Bigg]\\
    & \leq \Pr\Bigg[\frac{D_1(f_H||unif)}{D_1(f||unif)}\geq 1+t^\prime\Bigg]\\
    & \leq e^{-\frac{t^\prime n}{C\cdot (D_4(f||unif)+1)}}
\end{align*}
by Claim \ref{c4}.
\end{proof}

Now, we give a generalised version of Theorem \ref{5.1}. Note that this is a weak theorem and this stems from the fact that density functions very far from uniform are
not well-behaved with regards to  equator sampling.

\begin{theorem}
For density functions $f$ and $g$, $0<t<1$ and $\alpha>1$,
$$\Pr[|D_\alpha(\bar{f}_{|H}||\bar{g}_{|H})-D_\alpha(f||g)|\geq t]\leq Be^{-\frac{btn(\alpha-1))}{B(\alpha(D_2(f||unif)+1)+(D_{4\alpha}(f||g)+1)+(D_3(g||unif)+1))}},$$
for some constant $B,b>0$.
\end{theorem}
\begin{proof}
We know that with high probability, little re-normalization is required to make $f_{|H}$ and $g_{|H}$ density functions, and $f_{|H}$ is as close to the uniform distribution as $f$ by Claim \ref{c1} and \ref{c2}. In order to express $D_{\alpha}(h||unif)$ of Claim \ref{c2} in terms of $D_\beta(f||unif)$  and $D_\gamma(g||unif)$, for some constants $\beta,\gamma$, we do the following: since $h=\frac{\frac{f^\alpha}{g^{\alpha-1}}}{\int_{S^{n-1}}\frac{f^\alpha}{g^{\alpha-1}}dz}$,
\begin{align*}
    D_2(h||unif) & =\ln\Bigg[\int_{S^{n-1}\bigcap H}h^2dx\Bigg]=\ln\Bigg[\int_{S^{n-1}\bigcap H}\Bigg(\frac{f^{2\alpha}}{g^{2\alpha-2}}\Bigg)\cdot\Bigg(\frac{1}{\int_{S^{n-1}}\frac{f^\alpha}{g^{\alpha-1}}dz}\Bigg)^2dx\Bigg]\\
    & \leq \ln\Bigg[\int_{S^{n-1}\bigcap H}\Bigg(\frac{f^{2\alpha}}{g^{2\alpha-2}}\Bigg)dx\Bigg]\\
    & = \ln\Bigg[\int_{S^{n-1}\bigcap H}\frac{f^{2\alpha}}{g^{2\alpha-1/2}}\cdot g^{3/2}dx\Bigg]\\
    & \leq \ln\Bigg[\sqrt{\int_{S^{n-1}\bigcap H}\frac{f^{4\alpha}}{g^{4\alpha-1}}dx }\cdot \sqrt{\int_{S^{n-1}\bigcap H}g^{3}dx}\Bigg]\\
    & \leq D_4(f||g)+D_3(g||unif).
\end{align*}
Thus, following Claim \ref{c2},
\begin{equation}\label{D3D4}
    \Pr\Bigg[\Bigg|D_\alpha(f_{|H}||g_{|H})-D_\alpha(f||g)\Bigg|\geq t\Bigg]\leq e^{-\frac{tn(\alpha-1)}{C(D_{4\alpha}(f||g)+D_3(g||unif)+1)}}.
\end{equation}
Therefore,
\begin{align*}
    \Pr[|D_\alpha(f_{|H}||g_{|H})-D_\alpha(f||g)|\geq t] & \leq e^{-\frac{tn}{C\cdot D_2(f||unif)+C}}+e^{-\frac{tn}{C\cdot D_2(g||unif)+C}}+e^{-\frac{tn(\alpha-1)}{C(D_{4\alpha}(f||g)+D_3(g||unif))+C}}\\
    & \leq e^{-\frac{tn(\alpha-1)}{C(D_2(f||unif)+D_{4\alpha}(f||g)+2D_3(g||unif))+C}}
\end{align*}

We skip the part of the proof where $f_{|H}$ and $g_{|H}$ must be normalized.
\end{proof}

Note that this result is weak, in that it involves divergence of $f,g$ against the uniform distribution. This seems to be unavoidable.

\section{Concentration for Conditional R\'enyi Divergences}

The following lemma is easy to show and follows from our result below as a special case.

\begin{lem}\label{cor2}
Consider density functions $f_{AB},g_{AB}$ on a bipartite system $A\times B$. Let $f_A$ and $g_A$ be the marginal density functions on $A$.
$f_{A|b}$ and $g_{A|b}$ denote the normalized densities on $A$ when $b\in B$ is fixed. Then
$$\mathbb{E}_B[D_\alpha(f_{A|b}||g_{A|b})]\leq D_\alpha(f_{AB}||g_{AB}).$$
\end{lem}

Here the expectation is over $B$ according to $f_B$. This means the conditional  R\'enyi divergence is upper bounded by the total R\'enyi divergence.
This holds for all $\alpha\geq 1$.
We need a stronger result, in which we bound the expectation under a large enough event on $B$.

\begin{lem}\label{lalpha}
Let $E\subseteq B$ and let $f_B(E)=\int_B f_B(b)\cdot\mathbb{I}_E(b)\hspace{1mm}dB>0$, where $\mathbb{I}_E(b)$ is the indicator function defined  by:
$$\mathbb{I}_{E}(b)=\begin{cases} 1,\text{ if }b\in E\\
0,\text{ if }b\notin E\end{cases}.$$
Then the following holds for $\alpha>1$:
\begin{equation}\label{eqlalpha}
    \int_B f_B(b)\cdot\frac{\mathbb{I}_E(b)}{f_B(E)}\cdot D_\alpha(f_{A|b}||g_{A|b})\hspace{1mm}dB\leq \ln(g_B(E))- \frac{\alpha}{\alpha-1}\ln(f_B(E))+D_\alpha(f_{AB}||g_{AB})
\end{equation}
\end{lem}
\begin{proof}
By the definition of $D_\alpha(f_{A|b}||g_{A|b})$, $D_\alpha(f_{AB}||g_{AB})$, Jensen's inequality and the fact that $\mathbb{I}_E(b)\leq \textbf{1}_B$, where $\textbf{1}_B$ is the trivial all 1 function everywhere on $B$,
\begin{align*}
    & \int_B f_B(b)\cdot\frac{\mathbb{I}_E(b)}{f_B(E)}\cdot D_\alpha(f_{A|b}||g_{A|b})\hspace{1mm}dB\\
    & = \int_B f_B(b)\cdot\frac{\mathbb{I}_E(b)}{f_B(E)}\cdot \frac{1}{\alpha-1}\ln\Bigg(\int_A\Bigg(\frac{f_{AB}(a, b)}{f_B(b)}\Bigg)^\alpha\Bigg(\frac{g_B(b)}{g_{AB}(a, b)}\Bigg)^{\alpha-1}\hspace{1mm}dA\Bigg)dB\\
    & = \int_B f_B(b)\cdot\frac{\mathbb{I}_E(b)}{f_B(E)}\cdot \frac{1}{\alpha-1}\ln\Bigg(\Bigg(\frac{g_B(b)}{f_B(b)}\Bigg)^{\alpha-1}\cdot \int_A \frac{f_{AB}^\alpha(a, b)}{g_{AB}^{\alpha-1}(a, b)}\cdot \frac{1}{f_B(b)}\hspace{1mm}dA\Bigg)dB\\
    & = \int_B f_B(b)\cdot\frac{\mathbb{I}_E(b)}{f_B(E)}\cdot \frac{1}{\alpha-1}\ln\Bigg(\frac{g_B(b)}{f_B(b)}\Bigg)^{\alpha-1}\hspace{1mm}dB\\
    & + \int_B f_B(b)\cdot\frac{\mathbb{I}_E(b)}{f_B(E)}\cdot \frac{1}{\alpha-1}\ln\Bigg(\int_A\frac{f_{AB}^{\alpha}(a, b)}{g_{AB}^{\alpha-1}(a, b)}\cdot\frac{1}{f_B(b)}dA\Bigg)dB\\
    & \leq \ln\Bigg(\int_B \frac{\mathbb{I}_E(b)}{f_B(E)}\cdot g_B(b)\hspace{1mm}dB\Bigg)+ \frac{1}{\alpha-1}\ln\Bigg(\int_A\int_B f_B(b)\cdot\frac{\textbf{1}_B}{f_B(E)}\cdot\frac{f_{AB}^{\alpha}(a, b)}{g_{AB}^{\alpha-1}(a, b)}\cdot\frac{1}{f_B(b)}dA\Bigg)dB\\
    & = \ln\Bigg(\frac{g_B(E)}{f_B(E)}\Bigg)+D_\alpha(f_{AB}||g_{AB})+\frac{1}{\alpha-1}\ln\Bigg(\frac{1}{f_B(E)}\Bigg)\\
    & = \ln(g_B(E))-\ln(f_B(E))+D_\alpha(f_{AB}||g_{AB})-\frac{1}{\alpha-1}\ln(f_B(E))\\
    & = \ln(g_B(E))- \frac{\alpha}{\alpha-1}\ln(f_B(E))+D_\alpha(f_{AB}||g_{AB}).
\end{align*}
\end{proof}

We use this result now to show a concentration bound on the conditional R\'enyi divergence. By concentration here we mean an exponential bound on the upper tail of the distribution, not a concentration result around a fixed value.

\begin{corollary}\label{lalphac}
Let $E\subseteq B$ be defined as follows:
$$E=\{b\in B: D_\alpha(f_{A|b}||g_{A|b})\geq\ell\cdot D_\alpha(f_{AB}||g_{AB})\},$$
where $\ell$ is a constant. Then
$$f_B(E)\leq e^{-\frac{\alpha-1}{\alpha}(\ell-1)D_\alpha(f_{AB}||g_{AB})}.$$
\end{corollary}
\begin{proof}
From Lemma \ref{lalpha}, we have
\begin{align*}
\frac{\alpha}{\alpha-1}\ln(f_B(E)) & \leq -\int_B f_B(b)\cdot\frac{\mathbb{I}_E(b)}{f_B(E)}\cdot D_\alpha(f_{A|b}||g_{A|b})\hspace{1mm}dB+D_\alpha(f_{AB}||g_{AB})\\
& \leq-(\ell-1)D_\alpha(f_{AB}||g_{AB})
\end{align*}
Hence
$$f_B(E)\leq e^{-\frac{\alpha-1}{\alpha}(\ell-1)D_\alpha(f_{AB}||g_{AB})}.$$

We note here that we use that $\ln(g_B(E))\leq 0$.
\end{proof}

\begin{remark}
By a limiting argument, Lemma \ref{lalpha} and Corollary \ref{lalphac} hold for $\alpha=\infty$.
\end{remark}

First note that a stronger result holds regarding the ratio of $f_B(E)$ and $g_B(E)$. Furthermore no limiting argument yields anything interesting for $\alpha=1$. This is no coincidence, since it is easy to see that no interesting concentration result is true for the 1-divergence. The best one can get in full generality is the Markov bound.
Also, Lemma \ref{cor2} is still true for $\alpha=1$ by the chain rule.

\section{The Lower Bound on $R^{C\leftrightarrow B\to A}(aBc)$}

We prove the lower bound on the $aBc$ problem in a certain setting regarding the communication parameters: Charlie communicates at most $o(n^{2/3})$ bits, Bob at most $o(n^{1/3})$ bits and Alice not at all\footnote{She does have the last word, by deciding.}. In this situation we show the error to be constant. Our techniques do not allow us to prove stronger tradeoff lower bounds. Nevertheless an $\Omega(n^{1/3})$ lower bound for $aBc$ readily follows. The reason we cannot allow for a tradeoff in the lower bound is the way our different concentration bounds interact.

Assume that we have a communication protocol of the type Charlie $\leftrightarrow$ Bob $\rightarrow$ Alice, which partitions $O_n\times S^{n-1}$ into one-way rectangles. Let $M\subseteq O_n$ and $R\subseteq S^{n-1}$ be sufficiently large subsets, where $\mu(M)\geq e^{-\delta^2\cdot n^{1/3}}$ and $\mu(R)\geq 10e^{-\delta\cdot n^{2/3}}$ for small $\delta>0$. Define Alice's function $Acc$ on the rectangle with sides $M$ and $R$, as the map, $Acc:S^{n-1}\rightarrow\{-1, 1\}$. Recall that any efficient randomized protocol (in the above sense) leads to such a one-way rectangle with small error and similar size.

\begin{theorem}\label{lower}
Suppose we are given a communication protocol for $aBc$, where Bob and Charlie can send messages to each other and in the end send their transcript to Alice, who uses a function on her input and the received transcript to produce the output. If Bob communicates $\delta^2\cdot n^{1/3}$ bits and Charlie communicates $10\delta\cdot n^{2/3}$, then the communication protocol has error at least $\frac{1}{5}$ for some constant $\delta>0$.
\end{theorem}

\begin{proof}
Consider a one-way rectangle $M\times R\subseteq O_n\times S^{n-1}$ and a function $Acc:S^{n-1}\rightarrow \{-1, 1\}$ that determines whether to accept a particular $a$, given input sets $M$ and $R$. We partition $S^{n-1}$ into $L_{-1}$ and $L_1$ defined as follows:
$$L_{-1}=\{a\in L|Acc(a)=-1\}$$
$$L_1=\{a\in L|Acc(a)=1\}.$$

Without loss of generality, assume that $\mu(L_1)\geq 1/2$ and the average error on the defined inputs in $(L_1\bigcup L_{-1})\times M\times R$ is some $\epsilon>0$. Then, we remove from $L_1$ a size-$\mu(L_1)-1/2$ set of elements that have the largest average error among other elements in $L_1$ (to maintain the error after shrinking $L_1$). So, the new rectangle $L_1^\prime\times M\times R$ now has error at most $2\epsilon$.

Let $R^\prime=\{c \in R|D_1(\beta_c||unif)\geq \gamma\}$, where $\beta_c$ is the probability density function that arises when a random $B\in M$ is multiplied with a fixed $c\in R$ and a small constant $\gamma>0$. $R''$ is the same set for $-c$.

There are two cases: \\
\noindent \underline{Case 1: $\mu(R^\prime),\mu(R'')\leq\frac{1}{10}\mu(R)$}\\
In this case, we can remove $R^\prime$ from $R$ and let the resulting set be $\tilde{R}$. Now for all $c\in\tilde{R}$, we have $D_1(\beta_c||unif)\leq\gamma$ and the error of the rectangle $L_1^\prime\times M\times \tilde{R}$ is at most $\frac{20}{9}\epsilon$ after the removal of $R^\prime$ (at most $\frac{1}{10}$ of $R$) which causes the error to increase by a factor of $\frac{10}{9}$. Let $\tau(a)=\mathbb{E}_{c\in\tilde{R}}[\beta_c(a)]$. Then, $D_1(\tau||unif)\leq \gamma$.  By Pinsker's inequality, the following implication holds:

$$D_1(\tau||unif)\leq\gamma\Rightarrow |\int_{L_1'} \tau(a) d\mu -\int_{L_1'}unif(a) d\mu|\leq \delta(\tau, unif)\leq\sqrt{\frac{\gamma}{2}},$$
where $\delta$ is the total variational distance defined in subsection \ref{pinsker}.
Since $unif(a)=1$ for all $a$ and $\mu(L_1')= 1/2$, and hence when taking expectation over $L_1'$ the density for the expectation is 2 on every point in $L_1'$

$$1-\sqrt{2\gamma}\leq \mathbb{E}_{a\in L_1'}[\tau(a)]\leq 1+\sqrt{2\gamma}$$
and hence
$$\frac{1}{2}-\sqrt{\frac{\gamma}{2}}\leq\mathbb{E}_{a\in L_1'}\mathbb{E}_{c\in \tilde{R}}[\beta_c(a)]/2=\frac{\mu(L_1^\prime\times M\times \tilde{R}\bigcap H_1|H_1\bigcup H_{-1})}{\mu(L\times M\times R)}\leq \frac{1}{2}+\sqrt{\frac{\gamma}{2}}$$
$$\frac{1}{2}-\sqrt{\frac{\gamma}{2}}\leq\mathbb{E}_{a\in L_1'}\mathbb{E}_{c\in \tilde{R}}[\beta_{-c}(a)]/2=\frac{\mu(L_1^\prime\times M\times \tilde{R}\bigcap H_{-1}|H_1\bigcup H_{-1})}{\mu(L\times M\times R)}\leq \frac{1}{2}+\sqrt{\frac{\gamma}{2}},$$
where $H_1=\{(a, B, c)\in S^{n-1}\times O_n\times S^{n-1}|a^T\cdot B\cdot c=1\}$ and $H_{-1}=\{(a, B, c)\in S^{n-1}\times O_n\times S^{n-1}|a^T\cdot B\cdot c=-1\}$.
The second inequality follows by the same reasoning as the first via $R''$.

This means that $\mu(L_1^\prime\times M\times \tilde{R}\bigcap H_{1}|H_1\bigcup H_{-1})$ and $\mu(L_1^\prime\times M\times \tilde{R}\bigcap H_{-1}|H_1\bigcup H_{-1})$ are both within a $\bigg(\frac{1}{2}\pm\sqrt{\frac{\gamma}{2}}\bigg)$ factor of $\mu(L\times M\times R)$. However, $\mu(L_1^\prime\times M\times \tilde{R}\bigcap H_{-1}|H_1\bigcup H_{-1})\leq \frac{20}{9}\epsilon$, which is a contradiction since both $\gamma$ and $\epsilon$ are small. Therefore, it must be the case that
$$\frac{1}{2}-\sqrt{2\gamma}\leq \frac{20}{9}\epsilon\Rightarrow\epsilon\geq\frac{9}{40}-\frac{9}{20}\sqrt{2\gamma},$$
hence $\epsilon$ is large. Here setting $\gamma=10^{-3}$ works.
\vspace{3mm}

\noindent \underline{Case 2: $\mu(R^\prime)\geq\frac{1}{10}\mu(R)\geq e^{-\delta n^{2/3}}$ or the same for $R''$}\\

Without loss of generality we have the first condition true.
Given $R^\prime$ such that $\mu(R^\prime)\geq e^{-\delta n^{2/3}}$, we can find $\frac{n^{1/3}}{40\delta}$ orthogonal vectors in $R^\prime$ by Lemma 19 of \cite{caps}. Let these vectors be $c_1,\cdots, c_{\frac{n^{1/3}}{40\delta}}$ and extend to basis $c_1,\cdots,c_n$. Apply a unitary transformation $U$ that maps the $c_i$'s to $e_i$'s, where the $e_i$'s are the standard basis vectors in $\mathbb{R}^n$ with a 1 in the $i$-th entry and 0's elsewhere. This theorem will also hold for any other choice of basis. Let $M_i=M\cdot e_i$. The following is true by the chain rule:
\begin{align*}
    D_1(M||unif) & =\displaystyle\sum_{i=1}^n D_1(M_i|M_1,\cdots,M_{i-1}||unif|M_1,\cdots,M_{i-1})\\
    & \geq \displaystyle\sum_{i=1}^{\frac{n^{1/3}}{40\delta}} D_1(M_i|M_1,\cdots,M_{i-1}||unif|M_1,\cdots,M_{i-1}),
\end{align*}
where $D_1(P|A,B,\cdots||Q|A,B,\cdots)$ denotes the conditional divergence. We know that for all $c_i\in R^\prime$,
$$D_1(\beta_{c_i}||unif)=D_1(M_i||unif)\geq\gamma.$$
Assume throughout that $\mu(M)\geq e^{-\delta^2 n^{1/3}}$. Our goal is to show that
$$D_1(M)\geq \frac{n^{1/3}}{\delta}\cdot\Omega(\gamma),$$ which is equivalent to
$$\mu(M)\leq e^{-n^{1/3}\cdot const},$$
for a constant $const$, which will lead to a contraction to the above assumption.

Observe that
$$D_4(M_i|M_1,\ldots,M_{i-1}||unif)\leq D_4(M||unif)\leq \delta^2\cdot n^{1/3}$$
by Lemma \ref{cor2} and the data processing inequality.

Note that for all {\em fixed} $m_1,\cdots, m_{k-1}$,
\begin{align}\label{exp}
\begin{split}
    & D_1(M_i|m_1,\cdots,m_{k-1}||unif|m_1,\cdots,m_{k-1})\\
    & = \int_{S^{n-1}\cap m_1^\perp\cap\cdots\cap m_{k-1}^\perp} f_{M_i|m_1,\cdots, m_{k-1}}(x) \ln f_{M_i|m_1,\ldots, m_{k-1}}(x) dx\\
    & = \mathbb{E}_{v\in S^{n-1}\cap m_1^\perp\cap\cdots\cap m_{k-1}^\perp} \mathbb{E}_{x\in S^{n-1}\cap v^\perp\cap m_1^\perp\cap \cdots,\cap m_{k-1}^\perp} f_{M_i|m_1,\ldots, m_{k-1}} (x) \ln f_{M_i|m_1,\ldots, m_{k-1}}(x),\\
\end{split}
\end{align}
where the expectations are under the uniform distributions. Here $f_{M_i|\cdots}$ is normalized under the stated condition.

Note that we may apply our equator sampling bound from Theorem \ref{D1} for random $v$ from the sphere and get an upper bound on the deviation of $D_1(\bar{f}^v_{M_i|m_1,\ldots, m_{k-1}}||unif)$ from expectation showing that it is close to expectation, where the function here is the normalized conditional density function when fixing $v$, which is equal to $D_1(\bar{f}_{M_i|m_1,\ldots, m_{k-1},v}||unif)$, where $v$ is uniformly random but orthogonal to the span of the  $m_j$.

We say that $m_1,\cdots,m_{k-1}$ are "good" for $i$ for $k\leq i$ if the following criteria are satisfied.
\begin{itemize}
    \item $D_1(M_i|m_1,\cdots,m_{k-1}||unif|m_1,\cdots,m_{k-1})\geq (1-t)^{k-1} D_1(M_i||unif)\geq(1-t)^{k-1}\gamma$\\
    \item $D_4(M_i|m_1,\cdots,m_{k-1}||unif|m_1,\cdots,m_{k-1})\leq 2D_4(M||unif)\leq 2\delta^2\cdot n^{1/3}$\\
    \item $D_2(M_k|m_1,\cdots,m_{k-1}||unif|m_1,\cdots,m_{k-1})\leq 2D_2(M||unif)\leq 2\delta^2\cdot n^{1/3}$\\
\end{itemize}

We will show that the probability of the first criterion is at least $1-e^{-tn^{1/3}}$, whereas the probability of the second criterion is at least $1-e^{-\frac{3}{4} D_4(M||unif)}\geq 1-e^{-\frac{3}{4}\delta^2n^{1/3}}$ by Corollary \ref{lalphac} when we choose $\ell=2$ and $\alpha=4$ . The probability of the third criterion follows from Corollary \ref{lalphac} when we choose $\ell=\alpha=2$, which is at least $1-e^{-\frac{D_2(f||unif)}{2}}\geq 1-e^{-\delta^2n^{1/3}}$.

Now, fix any "good" $m_1\cdots m_{k-1}$. Then for $v\in S^{n-1}\cap m_1^\perp\cap\cdots \cap m_{k-1}^\perp$,
\begin{align}\label{prob}
\begin{split}
    & \Pr_v\bigg[\frac{D_1(M_i|m_1,\cdots, m_{k-1},v||unif)}{(1-t)^{k-1}D_1(M_i||unif)}\leq 1-t\Bigg]\\
    & \leq \Pr_v\bigg[\frac{D_1(M_i|m_1,\cdots, m_{k-1},v||unif|m_1\cdots m_{k-1},v)}{D_1(M_i|m_1,\cdots, m_{k-1}||unif|m_1,\cdots,m_{k-1})}\leq 1-t\Bigg] \\
    & \leq e^{-\frac{tn}{(1-t)^{k-1}D_4(M_i|m_1,\cdots ,m_k||unif|m_1\cdots m_k)}}\\
    & \leq e^{-\frac{tn}{(1-t)^{k-1}D_4(M||unif)}\cdot const}\\
    & \leq e^{-tn^{2/3}\cdot const}\\
\end{split}
\end{align}
by Theorem \ref{D1} and the data processing inequality. Here we use that $(1-t)^{k-1}$ is $\Omega(1)$, see below for a justification.

Let $E$ be the set defined as
$$E=\Bigg\{v\in S^{n-1}:(v\perp m_1)\wedge\cdots \wedge (v\perp m_{k-1})\wedge\Bigg(\frac{D_1(M_i|m_1,\cdots,m_{k-1},v||unif|m_1,\cdots,m_{k-1}, v)}{(1-t)^{k-1}D_1(M_i||unif)}\leq 1-t\Bigg) \Bigg\}.$$
Then, $\mu(E)\leq e^{-tn^{2/3}\cdot const}$ due to (\ref{prob}).  Define an indicator variable as follows:
$$\mathbb{I}_{E}(x)\begin{cases}
1,\text{if $x\in E$}\\
0,\text{if $x\notin E$}\\
\end{cases}.$$

Shortening $f_{M_k|m_1,\cdots,m_{k-1}}$ to $f$ we have by H\"older's inequality,
\begin{align*}
    f(E ) & =\int_{S^{n-1}} f(x)\cdot \mathbb{I}_E(x)dx\\
    & \leq \Bigg(\int_{S^{n-1}}(f(x))^\alpha\hspace{1mm}dx\Bigg)^{1/\alpha}\cdot \Bigg(\int_{S^{n-1}}(\mathbb{I}_E(x))^{1-1/\alpha}\hspace{1mm}dx\Bigg)^{1-1/\alpha}\\
    & = e^{\frac{\alpha-1}{\alpha}D_\alpha(f||unif)}\cdot\mu(E)^{(1-1/\alpha)},\\
\end{align*}
where $f_{M_k|m_1,\cdots,m_{k-1}}$ is the normalized density function of $M_k$ conditioned on $m_1,\cdots, m_{k-1}$. If $\alpha=2$, then
\begin{align*}
    f(E) & \leq e^{\frac{D_2(f||unif)}{2}}\cdot\mu(E)^{1/2}\\
    & \leq e^{\delta^2n^{1/3}}\cdot e^{-\frac{-tn^{2/3}\cdot const}{2}}\\
    & = e^{\delta^2n^{1/3}}\cdot e^{-\frac{-n^{2/3}\cdot \frac{\delta}{n^{1/3}}}{2}const}\\
    & = e^{\delta^2n^{1/3}}\cdot e^{-const\delta n^{1/3}}\\
    & = e^{-\mathcal{C}\delta n^{1/3}}\\
\end{align*}
when $t=\frac{40\delta}{n^{1/3}}$ and $\mathcal{C}$ is a constant. This is the probability of the first criterion of "good" $m_1,\cdots, m_{k-1}$.

Recall that $\mu(R^\prime)\geq e^{-\delta^2\cdot n^{1/3}}$ and for all $c_i\in R^\prime$, $D_1(\beta_{c_i}||unif)=D_1(M_i||unif)\geq\gamma$. We will be applying (\ref{prob}) repeatedly.

By Equation (\ref{exp}), we have for all $k\leq i-1$ and $i\leq t$ that
\begin{align*}
    & D_1(M_i|M_1,\cdots, M_{k}|||unif|M_1,\cdots, M_k)\\
    & =\mathbb{E}_{m_1,\cdots ,m_k} D_1(M_i|m_1,\cdots, m_k||unif|m_1,\cdots, m_k)\\
    & \geq \mathbb{E}_{\text{good }m_1,\cdots, m_{k}}\mathbb{E}_{m_k}D_1(M_i|m_1,\cdots, m_k||unif|m_1,\cdots, m_k)\cdot Prob(m_1,\ldots, m_k \mbox{good})\\
    & \geq D_1(M_i||unif)\cdot (1-t)^{k-1}\cdot (1-t)(1-e^{-\mathcal{C}n^{1/3}\cdot \delta})\cdot (1-(k-1)e^{-\mathcal{C}n^{1/3}\cdot \delta})\\
    & = D_1(M_i||unif)\cdot (1-t)^k\cdot (1-e^{-\mathcal{C}n^{1/3}\cdot \delta})\\
    & = D_1(M_i||unif)\cdot (1-t)^{1/t}\cdot (1-e^{-\mathcal{C}n^{1/3}\cdot \delta})\\
    & \geq \frac{D_1(M_i||unif)}{2e}\\
    & \geq \frac{\gamma}{2e}.\\
\end{align*}

\begin{comment}
After applying the probability bound (\ref{prob}) for $\frac{n^{1/3}}{40\delta}$ times, we get
$$\Pr\Bigg[D_1(M_i|M_1,\cdots,M_{i-1}||unif|M_1,\cdots,M_{i-1})\leq\frac{\gamma}{2e}\Bigg]\leq \frac{n^{1/3}}{40\delta}\cdot e^{-tn^{2/3}\cdot const}$$
by the union bound.
This implies that all $D_1(M_i|M_1,\cdots,M_{i-1}||unif|M_1,\cdots,M_{i-1})\leq \frac{\gamma}{2e}$ with high probability, for $i=1,\cdots \frac{n^{1/3}}{40\delta}$.

\end{comment}  Therefore,
$$D_1(M||unif)\geq\displaystyle\sum_{i=1}^{\frac{n^{1/3}}{40\delta}}D_1(M_i|M_1,\cdots,M_{i-1}||unif|M_1,\cdots,M_{i-1})\geq\frac{n^{1/3}}{40\delta}\cdot\frac{\gamma}{2e}\Rightarrow \mu(M)\leq e^{-n^{1/3}}$$

\end{comment}
By setting $\gamma=80\delta$ and the aforementioned $\gamma=10^{-3}$ we get our lower bound.\end{proof}
\begin{corollary}
$R^{C\rightarrow B\rightarrow A}(aBc)=\Omega(n^{1/3})$.
\end{corollary}
%%\begin{theorem}
%Let $M\subseteq O_n$ and $R\subseteq S^{n-1}$ have $\mu(M), \mu(R)\geq 2^{-\delta\cdot n^{1/3}} $ for the Haar measure. Denote by $\tau$ the density function of the probability distribution that %arises, when $B\in M$ and $c\in R$ are chosen uniformly from these sets independently, and the matrix-vector product $Bc$ is formed. Then for a constant $\gamma$,
%$$D_1(\tau||unif)\leq\gamma.$$
%\end{theorem}

\section{Conclusion}

Theorem \ref{lower} immediately shows us that no meaningful approximation of $a^TBc$ is possible in the streaming model with space complexity $o(n^{1/3})$. A stronger statement holds.
Even if $c$ is streamed $k$ times and $B$ is streamed $l$ times and $a$ is streamed just once (and last) then any streaming algorithm for $aBc$ needs space at least $\Omega(\min\{n^{2/3}/k, n^{1/3}/l\})$.

\section{Open Problems}

\begin{enumerate}
\item Show that $R(aBc)=\Omega(\sqrt n)$. This would be useful in the one-way setting and in general.
\item Show that there is a streaming algorithm that approximates $a^TBc$ within additive error $\epsilon$ with space $O(\sqrt n/\epsilon^2)$.
\item Show that $R(ABC)$ is large. Note that this requires that $n$ is even, because otherwise it is trivial to simply compute the determinant of $ABC$.
\item We conjecture that if one multiplies $c\in S^{n-1}\cap R$ and $B\in O_n\cap M$ for $M,R$ with $\mu(M),\mu(R)\geq 2^{-\delta\sqrt n}$, then the product $Bc$ is close to uniform on the sphere.
\item We also conjecture that if one multiplies $B,C$ uniformly at random from subsets $M,R$ of $SO_n$ such that $\mu(M),\mu(R)\geq 
2^{\delta\sqrt n}$ then the resulting distribution is close to uniform on $SO_n$.
\end{enumerate}

\bibliography{qc}
\bibliographystyle{plainurl}

\end{document}